\documentclass{sig-alternate}

\usepackage{etoolbox}
\makeatletter
\patchcmd{\maketitle}{\@copyrightspace}{}{}{}
\makeatother

\usepackage{yfonts}

\usepackage{amsmath, amssymb, enumitem, amsfonts,float,bm, stmaryrd}

\usepackage{graphicx,epsfig,epstopdf}

\usepackage{amsthm}
\newtheorem{theorem}{Theorem}

\newtheorem{definition}{Definition}

\usepackage{adjustbox}
\usepackage{color,graphicx}
\usepackage{url}
\usepackage{balance}
\usepackage{multirow,wrapfig}
\usepackage{rotating}
\usepackage{upgreek}
\usepackage{mathtools}
\usepackage{xspace}
\usepackage{style}
\usepackage{booktabs}
\usepackage{tikz}
\usetikzlibrary{arrows,backgrounds,decorations,decorations.pathmorphing,positioning,fit,automata,shapes,snakes,patterns}

\usepackage{prettyref}
\newcommand{\rref}[2][]{\prettyref{#2}}
\usepackage[ruled,linesnumbered]{algorithm2e}
\newrefformat{lemma}{Lemma\,\ref{#1}}
\newrefformat{fig}{Fig.\,\ref{#1}}
\newrefformat{mod}{Model \,\ref{#1}}
\newrefformat{line}{Line \,\ref{#1}}

\begin{document}



\renewcommand{\baselinestretch}{0.99}

\title{Forward Invariant Cuts to Simplify Proofs of Safety}

\numberofauthors{5}

\author{
%
%
\alignauthor
Nikos Ar\'echiga\\
       \affaddr{Carnegie Mellon University}\\
       \affaddr{Toyota InfoTechnology Center}\\
       \affaddr{Mountain View, CA}\\
       \email{\small{narechig@us.toyota-itc.com}}
\alignauthor
James Kapinski\\
       \affaddr{Toyota Technical Center}\\
       \affaddr{1630 W. 186th}\\
       \affaddr{Gardena, CA 90248}\\
       \email{\small{jim.kapinski@tema.toyota.com}}
\and
\alignauthor Jyotirmoy V. Deshmukh\\
       \affaddr{Toyota Technical Center}\\
       \affaddr{1630 W. 186th}\\
       \affaddr{Gardena, CA 90248}\\
       \email{\small{jyotirmoy.deshmukh@tema.toyota.com}}
\alignauthor Andr\'e Platzer\\
       \affaddr{Carnegie Mellon University}\\
       \affaddr{5000 Forbes Ave}\\
       \affaddr{Pittsburgh, PA}\\
       \email{\small{aplatzer@cs.cmu.edu}}
\alignauthor Bruce Krogh\\
       \affaddr{Carnegie Mellon University}\\
       \affaddr{5000 Forbes Ave}\\
       \affaddr{Pittsburgh, PA}\\
       \email{\small{krogh@rwanda.cmu.edu}}

}


\maketitle               

\begin{abstract} 
The use of deductive techniques, such as theorem provers, has several
advantages in safety verification of hybrid systems; however, 
state-of-the-art theorem provers require extensive manual intervention.
Furthermore, there is often a gap between the type
of assistance that a theorem prover requires to make progress on a
proof task and the assistance that a system designer is able to
provide.  
This paper presents an extension to \keymaera, a deductive
verification tool for differential dynamic logic; the new technique
allows {\em local reasoning} using system designer intuition about
performance within particular modes as part of a proof task.  
Our approach allows the theorem prover to
leverage forward invariants, discovered using numerical techniques,
as part of a proof of safety.  
We introduce a new inference rule into the proof calculus of  
\keymaera, 
the {\em forward invariant cut rule}, and we present a methodology 
to discover useful forward invariants, which are then used 
with the new cut rule to complete verification tasks. We 
demonstrate how our new approach can be used to complete 
verification tasks that lie out of the reach of existing deductive
approaches using several examples, including one involving an 
automotive powertrain control system.


\end{abstract}


\section{Introduction}\label{sec:intro}
Modern physical systems such as automobile
engines, avionics, and medical devices are controlled by software running on embedded
computing platforms.  In the software domain, techniques such as model
checking, theorem proving, and abstract interpretation have had
success verifying purely software systems. For physical systems,
techniques from dynamical systems theory and control theory such as
Lyapunov analysis have long been used to help characterize
system performance.  Most cyberphysical systems, however,
are {\em hybrid}, \ie, have both continuous state evolution governed
by differential equations and discrete mode transitions.
Most interesting analyses for such
systems (\eg, reachable set estimation) are undecidable
\cite{Henzinger1998}, and most software verification techniques are not
directly applicable.

Many extant approaches to hybrid system verification focus on creating an
overapproximation of the set of system states reachable over a fixed time
horizon
\cite{kapinski2004verifying,Frehse2008phaver,Frehse2011spaceex,chen2013}.  While
these approaches enjoy a high degree of automation, they are restricted in scope
and scalability.  Tools such as SpaceEx \cite{Frehse2011spaceex} and Flow*
\cite{chen2013} are susceptible to approximation error that worsens when the
reachable set estimation over continuous state-space interacts with discrete
switching, leading to false positives. The theorem prover PVS has been
used to reason about hybrid systems as composable hybrid automata in
\cite{PVS1} \cite{PVS2}. However, the continuous components are modeled by the
explicit solutions of the differential equations. Explicit solutions can
only be obtained for restricted classes of differential equations, e.g.
linear. On the other hand, d$\mathcal{L}$ allows reasoning about
continuous dynamics by using only the differential equations.


An alternative approach is to employ deductive techniques that attempt
to construct a symbolic proof of safety using a semi-interactive
theorem prover \cite{PlatzerLAHS}.
This approach has several advantages
in safety verification of hybrid systems. Unlike explicit reach-set computation
techniques, theorem provers can handle nonlinear dynamics directly, without
introducing approximation artifacts. Further, theorem provers can handle proof
tasks that involve symbolic parameters, with only the minimal constraints
required to guarantee safety. This makes the verification result reusable across
systems with parameter variations.
In the context of
dynamical or hybrid systems verification, a human may provide insight
to the theorem proving tool in the form of a {\em safety certificate},
\ie, a symbolic expression representing a set containing all reachable
states from a given initial set, while excluding unsafe states
\cite{beckert2007verification,PlatzerLAHS}. The tool can then
use this certificate to automatically prove system safety.

In \cite{PlatzerClarke2008}, the authors propose an approach that
begins with a global candidate certificate (in the form of a
{\em differential invariant}) that overapproximates
the reachable set of states. Constraints are iteratively added
until the overapproximation is small enough to exclude the unsafe set,
at which point the invariant becomes a safety certificate.  This
approach has had success in verifying aircraft roundabout maneuvers
using the \keymaera theorem prover. 
The notable aspects of this approach are: the initial input is in the
form of a {\em global certificate} of system safety, which is {\em
eagerly} constructed and then (globally) refined.

Cyberphysical system 
designs
have 
distinct modes of operation, with each mode 
corresponding to an (often) independently designed controller
operating regime.  Consequently, a designer has much more nuanced
information about mode-specific behaviors rather than overarching
knowledge about the entire system. The central thesis of this paper
is that when available, such additional information can be
useful for a theorem prover compared to a technique relying on
construction of a global safety certificate.  Our
approach encourages {\em local reasoning} and {\em lazy
construction} of 
certificates.

As an example of augmented local information, consider 
the scenario where a
designer knows that from a given set of modes, there are no discrete
transitions to unsafe system modes.  This is a form of
local certificate; in this case derived purely by reasoning over the
finite transition structure of the discrete modes.
Also consider the designer
insight that a system is expected to be stable in a certain mode. 
This is another form of local information that 
makes it possible to employ Lyapunov analysis-based techniques to obtain
a forward invariant set or barrier certificate that provides a
local certificate for that mode.


To support local reasoning, we introduce a new proof rule that we call
the {\em forward invariant cut rule} in the calculus of \keymaera.
Given a region of operation and a safe forward invariant for the
behaviors of that mode, the forward invariant cut rule allows us to decompose the
overall global safety proof into three proof obligations: (1) a proof
of invariance of the proposed certificate, (2) a proof that the
certificate guarantees safety, and (3) a proof of safety of everything
{\em but} the behaviors associated with the region covered by the certificate.
This makes it possible to carve out safe behavior and focus analysis only on the remaining part of the system.
An advantage of
the decompositional approach is that it allows us to defer the process
of producing a local certificate until we reach the relevant sub-goal
in the safety proof.  In other words, it allows {\em lazy}
construction of safe forward invariants, which is convenient as
certificates for 
system components are often easier to obtain than certificates for the aggregate system.



We demonstrate how our methodology can be used to complete verification
tasks that lie out of the reach of existing deductive techniques.
The systems we consider are hybrid and contain examples with continuous 
dynamical behaviors that are described by nonlinear ordinary differential 
equations (ODEs). Deductive approaches exist for addressing this class of 
systems, but the existing frameworks alone are insufficient to complete 
the proof tasks for the examples herein. For example, 
the framework in \cite{PlatzerClarke2008} provides a means to address the
examples we present using differential invariants, but the authors provide
no general method of computing the required differential invariant candidates.
Further, their technique requires reasoning about the global behavior of the
system, as opposed to the 
local invariant property that we require (a much weaker requirement).
%
The deductive proof system presented in \cite{dimitrova2014deductive}
uses local safety certificates to reason about behaviors but applies to
continuous (as opposed to hybrid) systems. Also, \cite{dimitrova2014deductive} provides no constructive means of
generating the necessary local safety certificates.
This is in contrast to our approach, which provides  methodologies for generating
the local safety certificates and including them in the proof task.

We present three examples that demonstrate the practical application 
of the forward invariant cut rule.
The first hybrid system is a hybrid system with three stable modes and one fail
mode.
The second system is a non-autonomous switched system, in which a user has the freedom to switch modes
at arbitrary instants.
The third system is a
simplified
model of an automotive subsystem that is responsible for
maintaining the air-to-fuel (A/F) ratio in an engine near an
optimal setpoint. In the automotive context, this is one of the most
important control problems with significant implications on fuel
efficiency and exhaust gas emissions. We are able to prove that the
A/F ratio remains within 10\% of the optimal setpoint value
using \keymaera.  


The paper is organized as follows. In \secref{prelim}, we
introduce the terminology and review material on hybrid programs (the
syntactic form used by \keymaera to express hybrid systems).  We
introduce the forward invariant cut rule in \secref{hybridcutsec}, and in
\secref{safetycertificates} we describe techniques for obtaining local
certificates.  We show how the forward invariant cut rule can be applied
to specific case studies in \secref{examples}. Finally, we conclude
and discuss related and future work in \secref{conclusions}.

\section{Hybrid systems and hybrid programs}\label{sec:prelim}
A hybrid system is a dynamical system with continuous-valued state
variables $\vx$ that take values from a domain $X\subseteq
\Reals^n$ and a discrete-valued state variable $q$ taken from a
finite set $Q$.  The system evolves in continuous or discrete time,
and the configuration of a hybrid system at time $t$ can be described
by the values of its continuous and discrete state variables.  The
discrete-valued states are called {\em modes} of operation.  The
hybrid state is given by the ordered pair $(\vx,q)\in X \times Q$.  In a
discrete mode $q$, the evolution of the continuous-valued state variables
is described by ordinary differential equations (ODEs)
\begin{eqnarray} \dot{\vx}(t)=f_q(\vx(t)), \end{eqnarray} where $f_q$
is a function from $X$ to $X$, often called the vector field. Though
hybrid systems are often described with external inputs, in this
paper we consider only {\em autonomous systems}, \ie, systems in
which all transitions depend only on the system states.  The
state-dependent conditions that allow the system to transition from
one discrete state to another (possibly same) discrete state are
called {\em guards}.

Hybrid systems are often modeled using hybrid automata.
We use \rref{fig:widget} as a running
example.  This example has four modes and two continuous-valued state
variables, with associated ODEs.  Modes are represented by nodes in
the graph; each mode $q$ has associated a unique set of ODEs ($f_q$).
There is a guard on the outgoing transition from $q_0$ to
$q_1$, and the transition from $q_0$ to $q_2$ is unguarded, so it can
always be taken.
The transition from $q_0$ to $q_2$ has a
nondeterministic reset allowing a jump from current state values $x_1$
and $x_2$ to any pair of values within the circle of radius two.  The
set of feasible initial conditions is indicated on the default
transition. Mode $q_0$ and $q_2$ have stable linear dynamics, and $q_1$
has stable nonlinear dynamics, as in Example 4.10 of \cite{Khalil2002}.

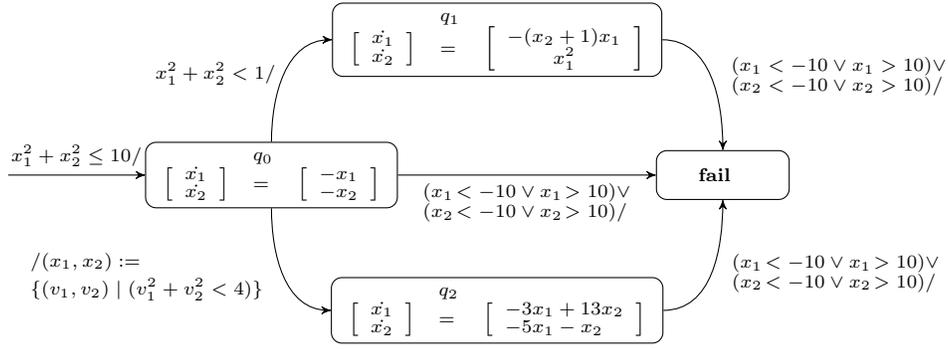
\begin{figure*}[t!]
\centering
\adjustbox{max width=\textwidth}{%
\begin{tikzpicture}[>=stealth']
\tikzstyle{smalltext}=[font=\fontsize{7}{7}\selectfont]
\tikzstyle{state}=[smalltext,rectangle,rounded corners,draw,minimum height=2em,inner sep=2pt,align=center]
\node[state] (mode1) {$\begin{array}{lll}
                       & q_1 & \\
                       \left[\begin{array}{l} 
                               \dot{x_1} \\
                               \dot{x_2}
                              \end{array}\right] & = &
                       \left[\begin{array}{c}
                                -(x_2 + 1) x_1 \\
                                x_1^2
                             \end{array}\right]
                       \end{array}$};

\node[coordinate,below of=mode1,node distance=18mm] (center) {};                        
\node[state,left of=center,node distance=30mm] (mode0) 
                      {$\begin{array}{lll}
                       & q_0 & \\
                       \left[\begin{array}{l} 
                               \dot{x_1} \\
                               \dot{x_2}
                              \end{array}\right] & = &
                       \left[\begin{array}{l}
                                -x_1 \\
                                -x_2 
                             \end{array}\right]
                       \end{array}$};
\node[state,right of=center,node distance=30mm,text width=5em] (modef)
                       {\textbf{fail}};

\node[state,below of=center,node distance=18mm] (mode2) 
                      {$\begin{array}{lll}
                       & q_2 & \\
                       \left[\begin{array}{l} 
                               \dot{x_1} \\
                               \dot{x_2}
                              \end{array}\right] & = &
                       \left[\begin{array}{l}
                                -3x_1 + 13x_2 \\
                                -5x_1 - x_2 
                             \end{array}\right]
                       \end{array}$};

\draw[->] (mode0) to[out=90,in=180] node[smalltext,left] {$x_1^2 + x_2^2 < 1 /$} (mode1);
\draw[->] (mode1) to[out=0,in=90] node[smalltext,right]  {$\begin{array}{l} 
                                                           (x_1 < -10 \vee x_1 > 10) \vee \\
                                                           (x_2 < -10 \vee x_2 > 10)/
                                                           \end{array}$} (modef);
\draw[->] (mode0) to[out=270,in=180] node[smalltext,left] {$\begin{array}{l}
                                                              \vspace{0.3em}
                                                              / (x_1,x_2) := \\ 
                                                              \{(v_1,v_2) \mid  (v_1^2 + v_2^2 < 4)\}
                                                            \end{array}$} (mode2);
\draw[->] (mode2) to[out=0,in=270] node[smalltext,right]  {$\begin{array}{l} 
                                                           (x_1\! < -10 \vee x_1\! > 10) \vee \\
                                                           (x_2\! < -10 \vee x_2\! > 10)/
                                                           \end{array}$} (modef);
\draw[->] (mode0) to node[smalltext,below] {$\begin{array}{l}                                                             
                                             (x_1\! < -10 \vee x_1\! > 10) \vee \\
                                             (x_2\! < -10 \vee x_2\! > 10)/            
                                            \end{array}$} (modef);

\node[left of=mode0,node distance=35mm,coordinate] (ghost) {};

\draw[->] (ghost) to node[smalltext,above] {$\begin{array}{l}
               x_1^2+x_2^2\leq 10 /                             
      \end{array}$} (mode0);

\end{tikzpicture}
}
\caption{A running example: All the modes have stable continuous dynamics, and there
is a special ``fail'' mode. \label{fig:widget}}
\end{figure*}



While hybrid automata are a convenient formalism, in this paper we
use the formalism of hybrid programs in order to facilitate
the use of the \keymaera theorem prover, which is the workhorse for
our deductive approach. Note that any hybrid automaton can be transformed into a hybrid program \cite{PlatzerLAHS}, therefore there is no loss of generality in considering hybrid programs.  \keymaera uses the formalism of
\emph{differential dynamic logic}, denoted by \dL.\footnote{ The syntax and semantics of \dL
are described in detail in \cite{PlatzerLAHS}; we provide only a
minimal overview here.}

\subsection{The logic \dL}
A {\em hybrid program} is specified by the grammar
\begin{align}
\label{eq:hp_syntax}
\alpha, \beta ::= ~&x\coloneqq \theta \mid
		x\coloneqq * \mid
		\{x_1' = \theta_1, \dots, x_n' = \theta_n\& H\} \\
		&\mid ?H \mid
		\alpha \cup \beta \mid \alpha;\beta \mid
		\alpha^*
\end{align}
where $\alpha, \beta$ are hybrid programs, $\theta, \theta_1, \dots,
\theta_n$ are terms, and $H$ is a logical formula. Intuitively, the
program $x \coloneqq \theta$ means that $x$ is assigned
the value of the term $\theta$.  The program $x \coloneqq *$ means
that $x$ is nondeterministically assigned an arbitrary
real value.  The program $\{x_1' = \theta_1, \dots, x_n' = \theta_n\&
H\}$ means that the variables $x_1, \dots, x_n$ evolve continuously
for some duration, with derivatives $\theta_1, \dots, \theta_n$,
subject to the constraint that $x_1, \dots, x_n$
satisfy $H$ during the entire flow.
The
hybrid program $?H$ behaves as a {\em skip} if the logical formula $H$
is true, and as an {\em abort} otherwise.

The nondeterministic choice $\alpha \cup \beta$ means that either
$\alpha$ or $\beta$ may be executed.  The sequential
composition $\alpha; \beta$ means that $\alpha$ is executed,
then $\beta$.  The nondeterministic repetition $\alpha^*$ means that
$\alpha$ is executed an arbitrary (possibly zero) number of times.
The logic \dL itself is a multimodal logic, in which the modalities
are annotated with hybrid programs. The formulas of \dL are described
by the grammar:
\begin{align}
\label{eq:dl_syntax}
\phi, \psi ::= ~&\theta_1 = \theta_2 \mid
		\theta_1 \ge \theta_2 \mid
		\neg \phi \mid
		\phi \wedge \psi\\
		& \mid \phi \vee \psi \mid
		\phi \rightarrow \psi \mid
		[\alpha] \phi \mid
		\left< \alpha \right> \phi
\end{align}
where $\phi$, $\psi$ are formulas of \dL, $\theta_1$, $\theta_2$ are
terms, and $\alpha$ is a hybrid program.  The box modality $[\alpha]
\phi$ means that $\phi$ holds after all traces of the hybrid program
$\alpha$, and $\left< \alpha \right>\phi$ means that $\phi$ holds
after some execution of hybrid program $\alpha$.

In the sequel, we will
abuse notation and use a formula interchangeably with the set that
it represents.


\subsection{Example}
\rref{mod:widget} shows a hybrid program representation of the running
example. \rref{line:widget_overview} shows how the subprograms are
assembled into the overall program. The system starts at a set
$I$ (\rref{line:widget_init}), and at
each iteration of the loop, one of the subprograms is
nondeterministically chosen for attempted execution. If the guard of
the subprogram succeeds, execution proceeds. The verification task is to show
that when this loop is executed any (finite)  number
of times, the state remains in the set $S$
(\rref{line:widget_safety}).
\rref{line:widget_mode0} is the guard and differential equations of $q_0$.
\rref{line:widget_s0_1} is the transition from $q_0$ to $q_1$ and the required guard.
\rref{line:widget_mode1_odes} proceeds
to specify the continuous evolution of $q_1$.
\rref{line:widget_s0_2_reset} applies the reset of the transition into $q_2$, which
indicates that the state resets anywhere in the circle of radius two.
\rref{line:widget_mode2_guard} checks the incoming guard to $q_2$
and \rref{line:widget_mode2_odes} specifies
the associated differential equations.
\rref{line:widget_s12_fail_guard}
specifies the guard that allows transitions into the failure mode. Note that the guard does not check the current
mode, since all of the modes may transition into the failure mode if the continuous states leave their prescribed bounds.
Line \ref{line:widget_sfail_fail_guard2} specifies that once
the failure mode is entered, it is not possible to leave it, and states $x_1, x_2$ maintain their previous values and do not evolve.

\SetAlgorithmName{Model}{model}{List of models}

\begin{algorithm}[h!]
\DontPrintSemicolon
\small{
\makebox[6em][r]{$\text{Ex}$} $\equiv$ $I \rightarrow [
				(\mathtt{m}_0 
				\cup \mathtt{s}_{0 \mapsto 1} 
				\cup \mathtt{m}_1 \cup$
                                $\mathtt{s}_{0 \mapsto 2} 
				\cup$ $\mathtt{m}_2$\;
				\makebox[7em][r]{}\hspace{1cm}$\cup \mathtt{s}_{\{0,1,2\} \mapsto fail} \cup$
                                $\mathtt{m}_{fail})^*] S$
	\nllabel{line:widget_overview} \;
\makebox[6em][r]{$I$} $\equiv$ $ x_1^2 + x_2^2 \le 10~\wedge~ M = q_0$
 	\nllabel{line:widget_init} \;
\makebox[6em][r]{$\mathtt{m}_0$} $\equiv$ $?(M = q_0);~\{x_1' = -x_1, x_2' = -x_2\}$
	\nllabel{line:widget_mode0} \;
\makebox[6em][r]{$\mathtt{s}_{0 \mapsto 1}$} $\equiv$ $?(M = q_0);~?(x_1^2 + x_2^2 < 1);~(M \coloneqq q_1);$
	\nllabel{line:widget_s0_1} \;
\makebox[6em][r]{$\mathtt{m}_1$} $\equiv $ $?(M = q_1);$
	\nllabel{line:widget_mode1_guard} \;
\makebox[7em][r]{} $\{ x_1' = -(x_2 + 1)*x_1,~x_2' = x_1^2 \}$
	\nllabel{line:widget_mode1_odes} \;
\makebox[6em][r]{$\mathtt{s}_{0 \mapsto 2}$} $\equiv$ $?(M = q_0);$
	\nllabel{line:widget_s0_2_guard}\;
\makebox[7em][r]{} $~x_1 \coloneqq *;~x_2 \coloneqq *; ?(x_1^2 + x_2^2 < 4);~(M \coloneqq q_2)$
	\nllabel{line:widget_s0_2_reset} \;
\makebox[6em][r]{$\mathtt{m}_2$} $\equiv$ $?(M = q_2);$
	\nllabel{line:widget_mode2_guard} \;
\makebox[7em][r]{} $\{x_1' = -3 x_1 + 13 x_2,~x_2' = -5 x_1 - x_2\}$
	\nllabel{line:widget_mode2_odes} \;
\makebox[6em][r]{$\mathtt{s}_{\{0,1,2\} \mapsto fail}$} $\equiv$ $?(-10 > x_1 ~\vee~x_1 > 10$\;
	\makebox[7em]{}$ ~\vee~-10 > x_2~\vee x_2 > 10);$
	\nllabel{line:widget_s12_fail_guard} \;
\makebox[7em][r]{} $M \coloneqq fail$
	\nllabel{line:widget_s12_fail_assign}\;
\makebox[6em][r]{$\mathtt{m}_{fail}$} $\equiv$ $?(M = fail);$
	\nllabel{line:widget_sfail_fail_guard2} \;
\makebox[6em][r]{$S$} $\equiv$ $M \neq fail$
	\nllabel{line:widget_safety} \;
}
\caption{Hybrid program for the running example}
\label{mod:widget}
\end{algorithm}

\section{Safety verification with the forward invariant cut rule}\label{sec:hybridcutsec}
\subsection{The safety verification problem}\label{sec:safety} 

The safety verification problem is to decide whether the state of a system
is always contained within a given safe set when starting from a designated
initial set, or equivalently, whether none of the behaviors enter an unsafe
set.

To formalize this problem in \dL, suppose $\alpha$ is a hybrid program
representation of the system of interest. Suppose $S$ is the safe set and
$I$ is the set of initial states.  Then the behaviors of $\alpha$ are
contained in $S$ if the following formula is a theorem of \dL.
\[
I \rightarrow [ \alpha^* ] S.
\]
The theorem prover \keymaera can be used to attempt to prove this.

To solve this problem, one might construct a set that contains all of the
system behaviors from the initial set and is contained in the safe set. We
call such a set a \emph{safety certificate}.  A safety certificate must
contain the initial state set, exclude the unsafe set, and be invariant for
system behaviors. We say that a set is \emph{initialized} if it includes
the initial set, \emph{safe} if it excludes the unsafe set, and
\emph{invariant} if whenever a system behavior enters it, the behavior
remains in the set for all future time.  Arguments with safety certificates
are captured in \dL using the invariant proof rule, where $C$ is a safety
certificate:
\[
\frac{I \rightarrow C~~~C \rightarrow[ \alpha ] C~~~C \rightarrow S}{I \rightarrow [ \alpha^* ] S }
\]
The general task of finding a safety certificate is difficult. In this
work, we propose instead a procedure that incrementally works towards a
proof.  Instead of a safety certificate, we use knowledge of system
structure to propose sets that are invariant and safe, but not necessarily
initialized, and leverage them in the proof procedure.

In our running example, modes $q_1$ and $q_2$ have stable dynamics. If a
Lyapunov function can be computed for either of these modes, its sublevel
sets (i.e., sets of the form $\{x ~|~ V(x) \le \ell\}$, for some $\ell
\ge 0$) will be invariant. The sublevel sets will be safe if they exclude the
transition to the fail mode, but they will not be initialized, since they
do not contain mode zero.

\subsection{The forward invariant cut rule}\label{sec:hybridcut}

A cut in a logical proof allows introducing a lemma. The main contribution
of this paper is a type of cut that simplifies the proof procedure by
leveraging knowledge of local invariance properties.

The following theorem establishes that if it can be shown that a predicate
($C$) is locally invariant ($C\rightarrow [\alpha]C$) and safe
($C\rightarrow S$), then the remaining conditions ($\neg C$) can be
separately addressed to prove safety. 

\begin{theorem}[Forward Invariant Cut Rule]
The following is a sound inference rule for the logic d$\mathcal{L}$.
\begin{equation}
\label{eq:hybrid_cut}
\frac{I \wedge \neg C \rightarrow [(\alpha ;? \neg C )^*] S~~~
C \rightarrow [\alpha] C~~~
C \rightarrow S}
{I \rightarrow [\alpha^*] S}
\end{equation}
\end{theorem}
\begin{proof}

We first provide a sketch in natural language. Let $\nu_0, \nu_1, \dots, \nu_n$
be any sequence of states of any length that are connected by runs of the hybrid
program $\alpha$.

{\bf Case a: } Suppose that none of the states in this sequence satisfy
$C$. Then this sequence is a run of the hybrid program $(\alpha; ?\neg
C)^*$, and is safe by the first premise.\label{casea}

{\bf Case b: } On the other hand, suppose $\nu_i \in C$ for some $0 \le i
\le n$.  Then the subsequence $\nu_i, \dots, \nu_n$ is a run of the program
$\alpha$ starting from $C$.  Then from the second and third premises of the
rule, $\nu_j \in C \subseteq S$ for all $j \ge i$. Note that the
subsequence $\nu_1, \dots, \nu_{i-1}$ is a run of program $\alpha^*$ such
that no state satisfies $C$, and is therefore safe by the previous case.

The formal proof follows.  Fix an interpretation $\mathcal{I}$ and an
assignment $\eta$.  From semantics of the second premise, if $\nu \in C$
and $(\nu, \omega) \in \rho_{\mathcal{I},\eta}(\alpha)$, then $\omega \in
C$.  From the semantics of the third premise, if $\omega \in C$, then
$\omega \in S$. From the semantics of the first premise, if $\nu \in I$ and
$\nu \notin C$, and $\omega$ is such that $(\nu, \omega) \in
\rho_{\mathcal{I},\eta}( (\alpha; ?\neg C)^* )$, then $\omega \in S$.  This
is equivalent to saying that for any $\omega$ such that there is a sequence
of states $\nu_0, \dots, \nu_n$, with $\nu_0 = \nu \in I \wedge \neg C$ and
$\nu_n = \omega$, $n \in \mathbb{N}$, and $(\nu_i, \nu_{i+1}) \in
\rho_{\mathcal{I},\eta}(\alpha; ?\neg C )$ for each $0 \le i \le n-1$, it
is the case that $\omega \in S$.

The proof is to show by induction that any state reachable by $\alpha^*$
from $I$ in $n \ge 0$ executions of $\alpha$ must be contained in $S$.  For
the base case, let $n = 0$. Then given $\nu \in I$, the only reachable
state by a sequence of length zero is $\nu$ itself. If $\nu \in C$, then
$\nu \in S$ by semantics of the third premise. If $\nu \notin C$, we have
that $(\nu, \nu) \in \rho_{\mathcal{I},\eta}( (\alpha; ?\neg C)^* )$ by a
chain of length zero, so that by semantics of the first premise, $\nu \in
S$.

As an inductive hypothesis, suppose that for every $\omega$ reachable by a
chain of length $n$, $\omega \in S$ (\ie, there exists $\nu_0, \dots,
\nu_n$ with $\nu_0 = \nu$ and $\omega = \nu_n$ such that $(\nu_i,
\nu_{i+1}) \in \rho_{\mathcal{I},\eta}(\alpha)$, for $0 \le i \le n-1$.
Now choose any state $\xi$ such that there is a chain of length $n+1$,
$\nu_0, \dots, \nu_{n+1}$ with $\nu_0 = \nu$ and $\nu_{n+1}=\xi$, such that
$(\nu_i, \nu_{i+1}) \in \rho_{\mathcal{I},\eta}(\alpha)$, for $0 \le i \le
n$).

First suppose that $\nu_n \in C$. Then by semantics of the second premise,
$\nu_{n+1} \in C$, and then $\nu_{n+1} \in S$ by semantics of the third
premise.  On the other hand, suppose $\nu_n \notin C$. We claim that for
all $j \le n$, $\nu_j \notin C$. To see this, note that if $\nu_j \in C$
for some $j \le n$, then $\nu_n \in C$ by semantics of the second premise,
which would contradict our assumption on $\nu_n$. Then we have that
$(\nu_i, \nu_{i+1}) \in \rho_{\mathcal{I},\eta}(\alpha; ?\neg C )$ for all
$0 \le i \le n$.  By semantics of the first premise, it follows that $\xi
\in S$.  This establishes the theorem.
\end{proof}

\subsection{Example}

For the running example, mode $q_1$ has a Lyapunov function of the form
$V_1(x_1, x_2)$ = $\frac{1}{2}x_1^2 + \frac{1}{2}(x_2 - 2)^2$ as described
in Example 4.10 of \cite{Khalil2002} (we discuss Lyapunov functions as
sources of invariants in \rref{sec:safetycertificates}). The sublevel set
$V_1(x_1, x_2) \le 5$ contains the reset into mode $q_1$. We apply the
forward invariant cut rule with $C_1 = V_1(x_1, x_2) \le 5 ~\wedge~M =
M_1$, a set that is invariant and safe, but not initialized since it does
not contain the initial mode $q_0$ of the hybrid system.  The rule
application causes the proof tree to split into three branches. The first
branch requires showing that whenever the system begins in $C_1$, it
remains in $C_1$. The only portions of the model that may run in this case
correspond to $q_1$ and the transition into the $\mathbf{fail}$ mode
(programs $\mathtt{m}_1$ and $\mathtt{s}_{0,1,2 \mapsto fail}$.  \keymaera
can readily check that since the proposed sublevel set excludes the guard
into $\mathbf{fail}$, $C_1$ will in fact be satisfied by the end of each
system trace. The second branch of the proof tree requires showing $C
\rightarrow S$, which is trivial, since $S$ is simply $M \ne M_{fail}$ and
$C$ stipulates $M = M_1$. We now turn our attention to the third branch.

Mode $q_2$ has a Lyapunov function $V(x_1, x_2)$ = $2 x_1^2 + 4 x_2^2$,
computed using standard Lyapunov techniques for linear systems. The
sublevel set $V_2(x_1, x_2) \le 16$ contains the circle of radius $2$; all
incoming transitions to mode $q_2$ make the system state to be reset to
somewhere within this circle.  By applying a forward invariant cut with
$C_2$ = $V_2(x_1, x_2) \le 16 ~\wedge~M = M_2$, we again get three
branches. As before, $C_2$ is invariant because the only portions of the
model that may run from $C_2$ are the programs $\mathtt{m}_2$ and
$\mathtt{s}_{0,1,2 \mapsto fail}$.  Since $V_2(x_1, x_2) \le 16$ excludes
the guard to $\mathtt{fail}$, \keymaera can show that $C_2$ represents a
safe set. The next branch is to prove that $C_2$ implies safety, which is
easy because $C_2$ requires $M = M_2$, which implies $M \neq M_{fail}$.

The third branch can now be easily proved with the standard tools of
\keymaera, using the loop invariant $M = M_0 ~\wedge~ x_1^2 + x_2^2 \le
10$.

\section{Obtaining safe forward invariants}\label{sec:safetycertificates}
This section describes various techniques to
generate safe forward invariants, which are invariant sets
that are safe but not necessarily initialized.
Let $\vx(t)$ denote any solution
trajectory for a given (hybrid) dynamical system.  A set $S$ is {\em
forward invariant} if for all $\vx(0) \in S$, for all $t$, $\vx(t) \in
S$. The general problem of identifying safe forward invariant sets that are useful is hard, but the techniques that we present can, in some cases, automatically identify safe forward invariant sets that can be used to complete safety proofs.

\subsection{Safe forward invariants based on Lyapunov analysis} \label{sec:safeinvs}

Lyapunov analysis provides one way to construct forward invariant sets for hybrid systems. We briefly review the basics of Lyapunov analysis to aid our
presentation.  Lyapunov's direct method is a well-known method used to
prove stability of dynamical systems within a region of interest. In
this method, the user provides a {\em local Lyapunov function}
$\lyap:X\to\Reals$ that over the domain of interest $X$ satisfies the
following properties:
\begin{enumerate}
\item Positive definiteness: for all $\vx$ in $X$, 
\begin{equation} \label{eq:lyapcond1}
\lyap(\vx) > 0,
\end{equation}
and $\lyap(\mathbf{0}) = 0$;

\item Derivative negative semidefiniteness: for all $\vx$ in $X$,
\begin{equation} \label{eq:lyapcond2}
\dot{\lyap}(\vx) = \frac{d}{dt}\lyap(\vx) \le 0,
\end{equation}
and $\dot{\lyap}(\mathbf{0}) = 0$.
\end{enumerate}

Existing techniques from dynamical systems theory use sum-of-squares
optimization \cite{parrilo2000SDP} and semidefinite programming
\cite{boyd94LMI,Boyd96} to identify {\em Lyapunov functions}
\cite{Meiss07} for systems described by polynomial differential
equations. A Lyapunov function $\lyap$ is analogous to a ranking function for
a discrete system, and it maps each continuous state $\vx$ to a positive
real number, with the property that along any system trajectory the
quantity $\lyap(\vx)$ monotonically decreases until it reaches $0$ at
the {\em equilibrium point}.  It is well-known that the {\em sublevel
set} of a Lyapunov function, $\levelset_\ell$ =
\mbox{$\setof{\vx|\lyap(\vx) \leq \ell}$} is a forward invariant set, \ie, given any initial condition in $\levelset_\ell$, all future
states remain in $\levelset_\ell$.  Thus, any sublevel set of a
Lyapunov function that includes the initial set and excludes the
unsafe set serves as a safety certificate
\cite{JohanssonRantzer1998,Prajna2005}.

\noindent {\em Remark:}
It is well known that for stable linear systems, a quadratic Lyapunov function
of the form $V = \vx^T P \vx$, where $P$ is a positive definite matrix, always
exists and can be computed by solving the matrix equation
\begin{equation} \label{eq:lyaplinear}
A^TP + P A = -Q
\end{equation}
where $Q$ is a positive definite matrix. Several scientific computing tools have
built-in commands to solve this equation, such as {\tt lyap} in \matlab and {\tt
LyapunovSolve} in \mathematica.

We now show how we can use Lyapunov-like functions to construct local certificates.

\mypara{Barrier Certificates} In the hybrid systems community, barrier
certificates have been
proposed as a Lyapunov-like analysis technique to prove that starting
from an initial set of states $X_0$, no system trajectory ever enters
an unsafe set $U$ \cite{Prajna2005,Prajna2006,prajna_safety_2004}.  The main step is to identify a barrier function
$B$ from the domain $X$ to $\Reals$, with the following properties:

\begin{eqnarray}\label{eq:barrier1}
\forall \vx\in X_0: B(\vx) \le 0  \\ \label{eq:barrier2}
\forall \vx\in U: B(\vx) > 0      \\ \label{eq:barrier3}
\forall \vx \in X\ \text{s.t.} B(\vx) = 0: \frac{\partial B}{\partial \vx}f(\vx) < 0.
\end{eqnarray}

Given a local Lyapunov function $\lyap$ valid in the domain $X$, 
if an $\ell$ can be selected such that (\ref{eq:barrier1}) and (\ref{eq:barrier2}) are satisfied, then $B(\vx)=\lyap(\vx)-\ell$
is a barrier certificate. This
follows from the definition of barrier certificates and the Lyapunov
conditions (\ref{eq:lyapcond1}) and (\ref{eq:lyapcond2}).

\mypara{Discovering Barrier Certificates} 
To discover barrier certificates, we employ a modification of a
technique from \cite{topcu08}, which uses concrete system
executions to generate a series of candidate Lyapunov functions.  Our
technique, which is based on \cite{Kapinski14},
 uses concrete executions to generate a set of linear
constraints.  A candidate Lyapunov function is then generated by
solving a linear program (LP) associated with the constraints.  A
series of candidates is iteratively improved upon, using a global
optimizer to search the region of interest for executions that violate
the condition (\ref{eq:lyapcond2}) for the given candidate.  The search is guided by
a cost function that is based on the Lie derivative of the candidate
Lyapunov function; if this cost function can be minimized below $0$,
then the minimizing argument provides a witness  (which we call a {\em
counterexample}) showing the candidate Lyapunov function is invalid.
Once such counterexamples are obtained, we include the associated
linear constraints in the LP problem and update the candidate Lyapunov  
function.  The process terminates when the global optimizer is unable
to find counterexamples to the candidate Lyapunov function. 
We then define the candidate barrier function $B(\vx)=V(\vx)-l$, where $l$ is selected such that 
(\ref{eq:barrier1}) and (\ref{eq:barrier2})
are satisfied.

Because there are no optimality guarantees from the global optimizer used to generate the candidate barrier function, the resulting candidate may not strictly satisfy the desired constraints. 
To check whether the candidates satisfy (\ref{eq:barrier1}) through (\ref{eq:barrier3}), we rely on a satisfiability modulo theories (SMT) solver that can handle nonlinear theories over the reals.
We use the dReal tool, which uses interval constraint
propagation (ICP) \cite{dReal}. dReal supports various nonlinear elementary functions in the framework of $\delta$-complete decision
procedures, and returns ``unsat'' or ``$\delta$-sat'' for a given query,
where $\delta$ is a precision value specified by the user. When
the answer is ``unsat'', dReal produces a proof of unsatisfiability; when it returns ``$\delta$-SAT'', it gives an interval of size $\delta$,
which contains points that may possibly satisfy the query.

When a ``$\delta$-SAT'' result is returned from a query to check (\ref{eq:barrier1}) through (\ref{eq:barrier3}), we do the following: 1.) construct a new linear constraint based on the interval returned, 2.)
add the new constraint to the existing set of linear constraints, and 3.) re-solve the LP to obtain an updated (improved) Lyapunov function candidate. If this process terminates, then the result is a barrier certificate.  

Our technique attempts to use discovered barrier certificates locally, that is, for each mode we attempt to construct a certificate that proves that the system will not leave the mode. If such a local barrier certificate is found, then the 
forward invariant cut rule can be applied to the mode to simplify the safety proof for the system, which may be composed of several modes.


\subsection{Other Techniques}

\mypara{Bounded-time Invariant Certificates}
Inspired by the success of reachability analysis using bounded model
checking for verifying software systems, there has been significant
research in estimating the reachable set of states for hybrid and
continuous-time dynamical systems.  See \cite{Frehse2011spaceex} and
references therein.  A common theme among various approaches is to
compute a {\em flowpipe}, or an overapproximation of the reachable
states over a bounded time horizon $\tau$. If the computed flowpipe
does not intersect with the unsafe set, then it is safe, and it is
invariant over bounded-time, as the initial states lie within it as wells as the set of all future states reachable within a fixed time bound also
lie within it.  The general form of a bounded-time invariant set is
$S_{reach} = (R(\vx) < 0) \wedge (t_l < t < t_u)$, where $R(\vx) < 0$
is some compact subset of the domain, and $(t_l,t_u)$ is the time
interval over which the set $R(\vx) < 0$ is invariant.

\mypara{Discrete Transition-based Certificates}
These certificates are useful to prove unreachability of certain modes
because of the transition structure of an underlying hybrid automaton.
Standard techniques from automata theory such as identifying strongly
connected components can be used to obtain such certificates.




\section{Case Studies}\label{sec:examples}
\newcommand{\startup}{\mathit{recovery}}
\newcommand{\normalmode}{\mathit{normal}}

\subsection{Non-autonomous Switched System} \label{sec:Academic2Dexample}


Consider an open two-mode system, where an external input can cause
the system to arbitrarily switch between the system modes.  This
example is significant, because neither of the two modes is invariant,
so the proof cannot rely on cutting out entire modes.  

The continuous dynamics are defined by matrices $A_1$ and $A_2$, as given below:
\[
\begin{array}{l}
\vspace{0.5em}
A_1 = \left[\begin{array}{ll}
        -1.0 & 4.0 \\
        -0.25 & -1.0 
        \end{array}\right], \\ 
A_2 = \left[\begin{array}{ll} 
         -1.0 & -0.25 \\
          4.0 & -1.0
          \end{array}\right].
\end{array}
\]
Linear reset maps are applied to the state when a transition is made between Modes 1 and 2. The resets are defined by matrices $R_{12}$ and
$R_{21}$:
\[
\begin{array}{l}
\vspace{0.5em}
R_{12} = \left[\begin{array}{ll} 
         -0.0658 & -0.0123 \\
          0.1965 & -0.0658 
         \end{array}\right], \\
R_{21} = \left[\begin{array}{ll} 
         -0.0658 & 0.1965 \\
         -0.0123 & -0.0658   
         \end{array}\right] .
\end{array}
\]

\newcommand{\bb}{\hspace{-0.2em}}
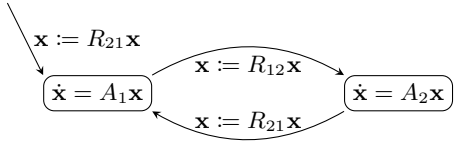
\begin{figure}[t]
\centering
\begin{tikzpicture}

\tikzstyle{smalltext}=[font=\fontsize{9}{9}\selectfont]
\node[rectangle,draw,rounded corners,smalltext] (s1) {$\dot{\vx} = A_1\vx$};
\node[rectangle,draw,rounded corners,right of=s1,node distance=40mm,smalltext] (s2) {$\dot{\vx} = A_2\vx$};
\node[coordinate,above of=s1,node distance=12mm] (c) {};
\node[coordinate,left of=c,node distance=12mm] (c1) {};
\draw[>=stealth,->] (c1) to 
                   node[smalltext,right] {$\vx \coloneqq R_{21}\vx$} 
                    (s1.north west); 
\draw[>=stealth,->] (s1.north east) to[out=30,in=150] 
                    node[smalltext,below] {$\vx \coloneqq R_{12} \vx$}
                    (s2.north west);
\draw[>=stealth,->] (s2.south west) to[out=210,in=330]
                    node[smalltext,above] {$\vx \coloneqq R_{21} \vx$}
                    (s1.south east);
\end{tikzpicture}
\caption{Hybrid automaton for the nonautonomous switched system.\label{fig:hans}}
\end{figure}

Figure \ref{fig:hans} shows a hybrid automaton for the system, and Model \ref{mod:ts} defines the corresponding hybrid program. For both modes, the continuous-time 
dynamics given by $A_1$ and $A_2$ are stable and linear.
It is well known that even for switched-mode systems with stable
linear continuous dynamics, switching conditions exists that lead to
instability for the switched system \cite{branicky1994}.  We wish to
prove that it is not possible to switch between $\vA_1$ and $\vA_2$ to
create unstable behavior.  The safety property for this system is
that it should remain within $\| \vx \|_\infty<2.0$. We apply the
forward invariant cut rule to the example to successfully prove the
safety property.  Below, we describe the steps of the proof.

\SetAlgorithmName{Model}{model}{List of models}
\DontPrintSemicolon
\begin{algorithm}[t!]
\small{
\makebox[5em][r]{TS} $\equiv$ $I \rightarrow [\left(\mathtt{s} \cup \mathtt{m}_1\; \cup \mathtt{m}_2 \right)^*] S$
\nllabel{line:ts_overview} \;

\makebox[5em][r]{I} $\equiv$ $M = 1 \wedge x_1^2 + x_2^2 \le 0.49$
\nllabel{line:ts_init}\;
\makebox[5em][r]{$\mathtt{s}$} $\equiv$ $ M \coloneqq 1 \cup M \coloneqq 2$
\nllabel{line:ts_switching}\;
\makebox[5em][r]{$\mathtt{m}_1$} $\equiv$ $(?M = 1);$
				\nllabel{line:ts_checkm1}\;
	\makebox[7em][r]{}$x_1 \coloneqq -0.0658 x_1+0.1965 x_2;$
				\nllabel{line:ts_reset_m1_1}\;
	\makebox[7em][r]{}$x_2 \coloneqq  -0.0123 x_1-0.0658 x_2;$
				\nllabel{line:ts_reset_m1_2}\;
	\makebox[7em][r]{}$\{x_1' = -x_1 + 4 x_2, x_2' = -(1/4) x_1-x_2\}$
				\nllabel{line:ts_dynamics_m1}\;

\makebox[5em][r]{$\mathtt{m}_2$} $\equiv$ $(?M = 2);$
				\nllabel{line:ts_checkm2}\;
	\makebox[7em][r]{}$x_1 \coloneqq -0.0658 x_1-0.0123 x_2 $
				\nllabel{line:ts_reset_m2_1}\;
	\makebox[7em][r]{}$x_2 \coloneqq 0.1965 x_1-0.0658 x_2 $
				\nllabel{line:ts_reset_m2_2}\;
	\makebox[7em][r]{}$\{x_1' = -x_1-(1/4) x_2, x_2' = 4 x_1-x_2 \}$
				\nllabel{line:ts_dynamics_m2}\;
\makebox[5em][r]{S} $\equiv$ $x_1 > -2 \wedge x_1 < 2 \wedge x_2 > -2 \wedge x_2 < -2$
	\nllabel{line:ts_safety}
\caption{A d$\mathcal{L}$ model of the nonautonomous switched system\label{mod:ts}}
}
\end{algorithm}

Here, the designer provided two forward invariants of the system by
independently solving the Lyapunov equation (\ref{eq:lyaplinear}) for
the linear dynamics of the system in each of the modes. The designer
then picked level set sizes to ensure that the resulting forward
invariant is contained within the safe set $S$. The invariants are
given below:

\begin{eqnarray}
C_1 = \{ \vx \mid \lyap_1(\vx) < l_1 \} 
\label{eq:first_cut} \\
C_2 = \{ \vx \mid \lyap_2(\vx) < l_2 \}
\label{eq:second_cut}
\end{eqnarray}

Here, $\lyap_1(\vx) = 0.3828 x_1^2 + 0.9375 x_1 x_2+2.3750 x_2^2$, and
$l_1 =  1.0$, and $\lyap_2(\vx) = 2.3750 x_1^2 + 0.9375 x_1 x_2 + 0.3828
x_2^2$, and $l_2 = 1.0$.

We sequentially apply two forward invariant cuts in order to prove
Model~\ref{mod:ts} safe. The first forward invariant cut rule uses the set
$C_1$ as the cut. After applying
$C_1$, the proof tree has three branches: $I \wedge \neg C_1 \rightarrow
[(\alpha; ?\neg C_1)^*] S$, $C_1
\rightarrow [\alpha] C_1$, and $C_1 \rightarrow S$. Of these, the third
branch is trivially true as $C_1 \subseteq S$. To prove the second branch
valid, \keymaera needs to prove that $C_1$ is invariant for the disjuncts.

For the hybrid program $\mathtt{m}_1$, \keymaera computes the forward image
of the set $C_1$ when transformed by the linear transformation $R_{21}$,
\ie, the set $F = \{ \vy \mid \vy = R_{21} \vx \wedge \lyap_1(\vx) < l_1
\}$. Note that this step requires performing quantifier elimination, and
\keymaera utilizes Mathematica for this purpose.  It then uses $C_1$ as a
differential invariant to prove that $F \rightarrow [\{\vx'= A_1 \vx\}]
C_1$.  This is facilitated by the fact that $C_1$ is in fact invariant for
the linear system $\dot{\vx} = A_1\vx$.

The difficult branch is the one requiring us to prove that $C_1$ is
invariant for mode $\mathtt{m}_2$. To do so, we assist \keymaera with
certain lemmas; the intuition for these lemmas is as follows:  Any state in
set $C_1$ upon executing the program $\mathtt{m}_2$ is linearly
transformed by $R_{12}$.  Let $\hat{C}_1 = \{\vxhat \mid \vx\in C_1
\wedge \vxhat
= R_{12} \vx\}$ represent the forward image of $C_1$ under $R_{12}$.  Next,
we show that the set $\hat{C}_1$ is a subset of a specific sublevel set
$C_2^*$ of $\lyap_2(\vx)$.  As $C_2^*$ is a sublevel set of $\lyap_2(\vx)$,
it is invariant under the dynamics $\dot{\vx} = A_2 \vx$; thus, any state
beginning in $C_2^*$ will remain in $C_2^*$. Finally, we choose $C_2^*$ in
such a way that $C_2^* \subseteq C_1$. This essentially proves that any
state starting in the set $C_1$ will be contained in set $\hat{C}_1$, of
which any state will under the dynamics $\dot{\vx} = A_2 \vx$ remain in the
set $C_2^*$, \ie, in the state $C_1$. 

Formally, we establish the following:
\begin{eqnarray} 
        C_1 \rightarrow [\vx := R_{12}\vx] \hat{C}_1 \\
        \hat{C}_1 \subseteq C_2^* \\
        C_2^* \rightarrow [\{\vx' = A_2\vx\}] C_2^* \\
        C_2^* \subseteq C_1
\end{eqnarray}

We can combine these to infer that $C_1 \rightarrow [\mathtt{m}_2] C_1$.

Finally, the first branch of the proof considers $I \wedge \neg C_1$; this
contains the set of initial states not in $C_1$. These can now be addressed
by the second forward invariant cut (set $C_2$) following a symmetric
argument as above. After applying the second cut $C_2$, the first branch
has an empty antecedent ($I \wedge \neg C_1 \wedge \neg C_2$ is empty),
\ie, the proof has accounted for all initial states, which closes the
proof. The sets we have discussed are shown in Figure
\ref{fig:switchedmode}.

\begin{figure}
\begin{center}
\includegraphics[clip=true,trim=0in 0in 0in 0in, width=3.55in]{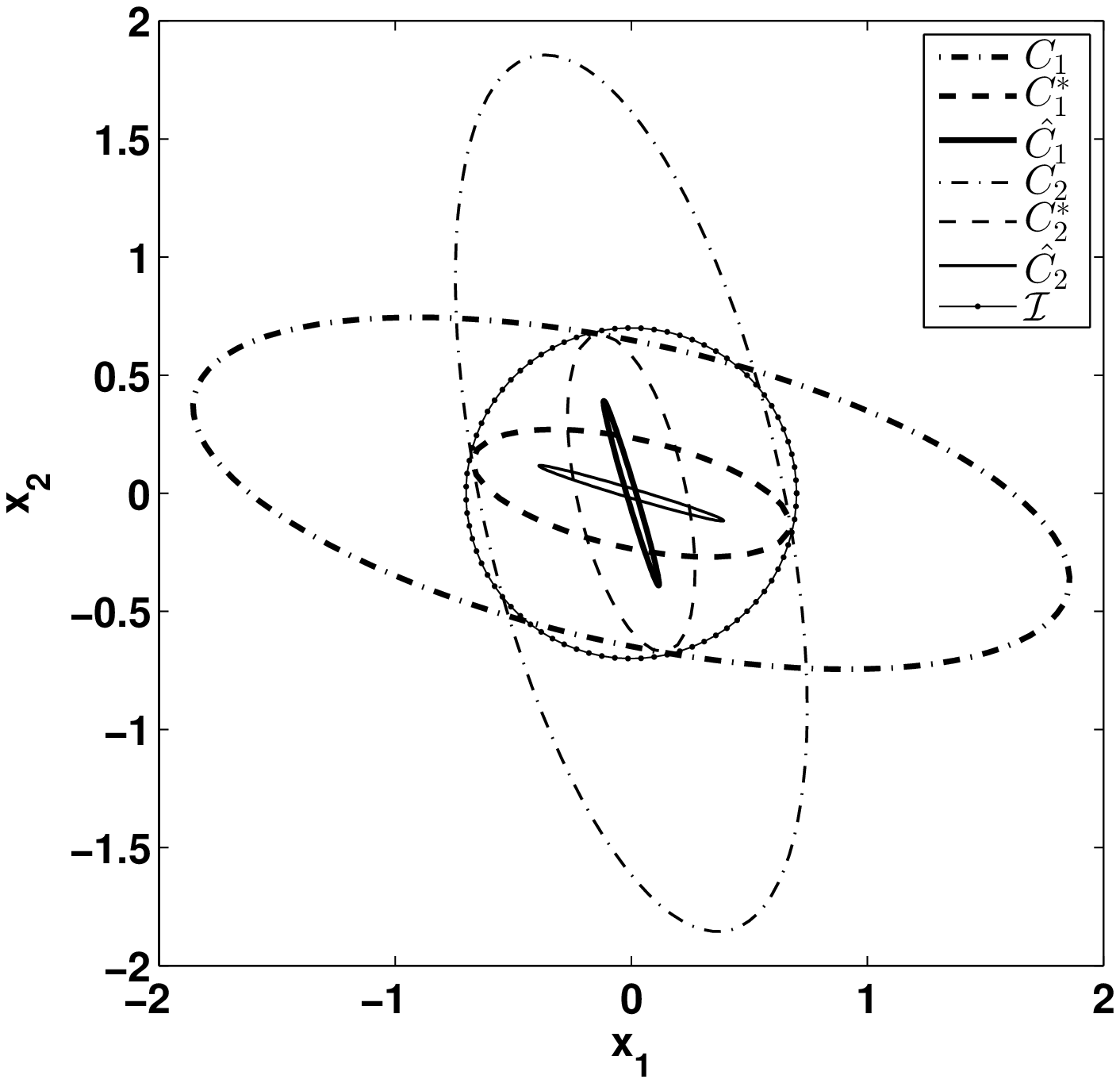}
\caption{Illustration of forward invariant sets for Example \ref{sec:Academic2Dexample}. \label{fig:switchedmode}}
\end{center}
\end{figure}


\subsection{Engine fuel control}\label{sec:AFCexample}

\mypara{Model} Our second case study is a hybrid system representing an
automotive fuel control application.   Environmental concerns and
government legislation require that the fuel economy be maximized and the
exhaust gas emissions (e.g., hydrocarbons, carbon monoxide, and nitrogen
oxides) be minimized.  At the ideal air-to-fuel (A/F) ratio, also known as
the {\em stoichiometric} value, both these quantities are optimized.  We
present an automotive control system whose purpose is to accurately
regulate the A/F ratio.

The system dynamics and parameters were derived from a published model
\cite{Jin14} and then simplified, as in \cite{Kapinski14}.  The model
consists of a simplified version of the physics of engine subsystems
responsible for air intake and A/F ratio measurement, along with a computer
control system tasked with regulating the A/F ratio.  The objective of the
controller is to maintain the A/F ratio within 10\% of the nominal
operating conditions.  The experiment that we model involves an engine
connected to a dynamometer -- a device that can control the speed of the
engine and measure the output torque. In our setting, the dynamometer
maintains the engine at a constant rotational velocity.  The controller has
two modes of operation: (1) a $\startup$ mode, which controls fuel in an
open-loop manner, \ie, with only feedforward control action, where the
system runs for at most $8$ms, and (2) a normal run mode, which uses
feedback control to regulate the A/F ratio.  

The controller measures both the air flow through the air-intake manifold,
which it uses to estimate the air pressure in the manifold, and the oxygen
content of the exhaust gas, which it uses to compute the A/F ratio.  The
$\startup$ mode represents the behavior of the controller when recovering
from a sensor fault (e.g., aberrant sensor readings, environmental
conditions that cause suspicion of the sensor readings). During the
$\startup$ mode, the controller has no access to oxygen sensor measurements
and so must operate in a feedforward manner (i.e., using only the manifold
air flow rate). The normal mode is the typical mode of operation, where the
oxygen sensor measurements are used to do feedback control.

Model \ref{mod:efc_hybridprogram} is a hybrid program representing this
system. The ODEs representing the continuous dynamics in each mode and
the model parameters are presented in the Appendix. The state variables
$\phat,\rhat,\pesthat$, and $\ihat$ represent the manifold pressure, the
ratio between actual air-fuel ratio and the stoichiometric value, the
controller estimate of the manifold pressure, and the internal state of
the PI controller; these variables have all been translated so that the
equilibrium point coincides with the origin. In the $\startup$ mode, the
continuous-time state $\vx$ is the tuple
$(\phat,\rhat,\pesthat,\ihat,\tau)$.  The additional state variable in
the $\startup$ mode represents the state of a timer that evolves
according to the ODE $\dot{\tau} = 1$.  In the $\normalmode$ mode, the
state is given by $(\phat,\rhat,\pesthat,\ihat)$.

We assume the system is within $1.0\%$ of the nominal value at the
initialization of the $\startup$ mode. This represents the case where the
system was previously in a mode of operation that accurately regulated the
A/F ratio to the desired setpoint.  A domain of interest for the state
variables is given by $\|\vx \|_\infty<0.2$.

\SetAlgorithmName{Model}{model}{List of models}
\DontPrintSemicolon
\begin{algorithm}[t!]
\small{
\makebox[5em][r]{EFC} $\equiv$ $I \rightarrow [\left(\mathtt{m}_1
                                                              \cup \mathtt{m}_2
                                                              \cup \mathtt{s}_{1 \mapsto 2}
                                                              \cup \mathtt{s}_{\{1,2\} \mapsto fail}
                                                              \cup \mathtt{m}_{fail}
                                                              \right)^*] S$ \nllabel{line1}
                                                              \;
\makebox[5em][r]{$I$} $\equiv$ $(-0.001 \le p \le 0.001)~$\;
\makebox[6em][r]{} $\wedge~(-0.001 \le r \le 0.001)~\wedge$ \;
\makebox[6em][r]{} $\wedge~(p_{est} =  0 \wedge i = 0 \wedge M = 1)$ \;
\makebox[5em][r]{$\mathtt{m}_1$} $\equiv$ $(?M = \startup; ?\tau \le 0.008;$ \;
\makebox[6em][r]{} $\{\exists \ell_1.\exists \ell_2. \exists \ell_3$\;
\makebox[7em][r]{} $(-0.86 \le \ell_1 \le 0.74)$\;
\makebox[7em][r]{} $(-0.17 \le \ell_2 \le 0.18)$\;
\makebox[7em][r]{} $(-0.81 \le \ell_3 \le 0.68)$\;
\makebox[7em][r]{} $\wedge~(p' = \ell_1)~\wedge~(r' = \ell_2)~\wedge~(p_{est}' =
\ell_3) ~\wedge$ \nllabel{line3} \;
\makebox[7em][r]{} $i' = 0 ~\&~ \tau' = 1 ~\wedge~ \tau \le 0.008\}$\;
\makebox[5em][r]{$\mathtt{s}_{1 \mapsto 2}$} $\equiv$ $(?M = \startup;?\tau \ge 0.008;$ \;
\makebox[6.5em][r]{} $M \coloneqq \normalmode;)$\;
\makebox[5em][r]{$\mathtt{m}_2$} $\equiv$ $(?M = \normalmode;$ \nllabel{line:m2}\;
\makebox[7em][r]{}$\{ p' = f_p,$\;
\makebox[7em]{}$r' = f_r,$\;
\makebox[7em]{}$p_{est}' = f_{p_{est}},$\;
\makebox[7em]{}$i' = f_i,$\;
\makebox[7em]{}$\& -0.02 \le p \le 0.02~\wedge~-0.02 \le r \le 0.02$\;
\makebox[7em]{}$\wedge -0.02 \le p_{est} \le 0.02~\wedge~-0.02 \le i \le 0.02\}$\;
\makebox[5em][r]{$\mathtt{s}_{\{1,2\}\mapsto fail}$} $\equiv$ $(?(r < -0.1 \vee r > 0.1);$ \nllabel{line:s12fail}\;
\makebox[7em][r]{} $M \coloneqq fail)$ \;
\makebox[5em][r]{$\mathtt{m}_{fail}$} $\equiv$ $(?(r < -0.1 \vee r > 0.1 );$ \nllabel{line:mfail}\;
\makebox[7em][r]{} $M \coloneqq fail)$ \;
\makebox[5em]{$S$} $\equiv$ $M \ne fail$ \;
\caption{A d$\mathcal{L}$ model of a closed-loop fuel control system\label{mod:efc_hybridprogram}}
}
\end{algorithm}

\mypara{Safety proof using forward invariant cut} The verification goal is
to ensure that in the given experimental setting, the system always remains
within 10\% of the nominal A/F ratio after a fixed recovery time of $0.8$
ms has passed. In other words, we wish to show that the system begins in
the $\startup$ mode, with the initial set of continuous states defined by
$\mathtt{init} = \{ \vx \mid \|\vx\|_\infty < 0.01 \}$; the system
transitions to the $\normalmode$ mode after at most $8.0$ ms; and the
system never transitions to the unsafe set, where $|r| > 0.1$, within the
domain of interest $\|\vx \|_\infty<0.2$.  

%


In previous work \cite{Kapinski14}, the authors had established a forward
invariant set for the $\normalmode$ mode of operation using a barrier
certificate formulation. The authors formulated the barrier certificate
using simulation-guided techniques to obtain a candidate Lyapunov function
$\lyap$ and a number $\ell$ to propose a barrier function of the form
$B(\vx) = \lyap(\vx) - \ell$. Here, $\lyap(\vx)$ = $\vz^T\vP\vz$, and $\vz$
is a vector of all monomials of degree $\le 2$ of the state variables
$\phat$, $\rhat$, $\pesthat$ and $\ihat$.  Note that $\vz$ thus contains
$14$ monomials, and $\vP$ is a $14\mathrm{x}14$ matrix.  We omit the
resulting $\vP$ matrix for brevity.

We use the set enclosed by the barrier function to formulate the forward
invariant cut
\begin{equation}\label{eq:hybridcut}
C \equiv (M = \normalmode) \wedge (B(\vx) \le 0).
\end{equation}

Application of the forward invariant cut inference rule
\eqref{eq:hybrid_cut}, generates three proof obligations that \keymaera has
to discharge.

\mypara{Obligation 1} $C \rightarrow [\alpha] C$ \\ Note that once we
define $C$, the hybrid programs $\mathtt{m}_1$, $s_{1 \mapsto 2}$ can be
excised by \keymaera, as both have the hybrid program $?M=\startup$ as
their first item, which is inconsistent with $C$. Thus, \keymaera can then
focus on proving this obligation only for the programs $\mathtt{m}_2$,
$\mathtt{s}_{\{1,2\} \mapsto fail}$ and $\mathtt{m}_{fail}$.

In order to discharge the obligation for the program $\mathtt{m}_2$, we
first perform some trivial simplifications with \keymaera that leaves us
with the following proof goal:

\begin{equation}
\label{eq:prebarr}
(B(\vx) \le 0) \rightarrow [ \{ \vx' = f(\vx) \& H \} ] (B(\vx) \le 0) \wedge (M = \normalmode) 
\end{equation}

To discharge \eqref{eq:prebarr}, we can use the barrier certificate rule
shown in \eqref{eq:keymaera_barrier} that we have added to KeYmaera's proof
calculus.

\begin{equation}\label{eq:keymaera_barrier}
\frac{ \mathit{init}\!\rightarrow\! B(\vx)\!\le\!0\quad  %
       B(\vx)\!=\!0\!\rightarrow\!\frac{\partial B}{\partial \vx}\!\!\cdot\!\! f(\vx)\! <\!0  \quad %
       B(\vx)\! \le\! 0\!\rightarrow\! \mathit{safe} %
}{ %
\mathit{init} \rightarrow[ \{ x' = f(\vx) \}  ] \mathit{safe}
}
\end{equation}
where $H$ is the domain of evolution of the continuous dynamics.
In our application of the barrier certificate rule, we substitute
$\mathit{init}$ with $(B(\vx) \le 0)$ and $\mathit{safe}$ with $(B(\vx) \le
0) \wedge (M = \normalmode)$.  The first and the third proof obligations in
the barrier certificate rule are then trivially satisfied.  For the
remaining (middle) proof obligation \keymaera uses the SMT solver dReal
\cite{dReal}. In particular, it asks dReal if the query $(B(\vx) = 0)$
$\wedge$ \mbox{$(\frac{\partial B}{\partial\vx}\cdot f_{normal}(\vx) >
-\epsilon)$} is unsatisfiable, where $\epsilon$ is a small positive number.


In order to discharge the proof obligation for $\mathtt{m}_2$,
$\mathtt{s}_{\{1,2\} \mapsto fail}$, \keymaera needs to show that
if $B(\vx) < 0$ holds, either of these programs cannot invalidate $C$
by transitioning to mode $fail$. It proves this by showing that
the set $B(\vx) < 0$  is a subset of the safe set using dReal.

\mypara{Obligation 2} $C \rightarrow S$ \\ This obligation is
trivial as $S$ requires the mode to be $fail$, while $C$ says that the
mode is $\normalmode$ mode.

\mypara{Obligation 3} $I \wedge \neg C \rightarrow$ $[(\alpha;?\neg C)^*] S$  \\

To prove this obligation, we use the lemma that the set $C1$ is an
invariant for all states remaining in $I \wedge \neg C$ . This is a
bounded-time invariant certificate.

\begin{align}
C1 \equiv (M \neq \mathtt{m}_{fail} ) \wedge (0 \le \tau \le 0.008) \wedge (x \in S_{reach}) \label{eq:hybridcut2}
\end{align}

Here $S_{reach}$ is an overapproximation of reachable sets by using upper
and lower bounds on $\dot{p}$ and $\dot{r}$ computed using dReal. The proof
for this branch continues using standard KeYmaera deduction procedures.
There is one additional barrier certificate application to show that the
normal mode, when starting from this set, lands within the barrier
certificate and therefore also respects this invariant. This requires a
derivative negativity argument, which \keymaera again handles via an
external dReal query.



\section{Related Work and Conclusions}\label{sec:conclusions}
\mypara{Lazy abstraction} In software verification using conservative
abstractions, an abstract program can be viewed as a proof of program
correctness if it satisfies the correctness property of interest. A
popular paradigm is that of {\em lazy abstraction}
\cite{henzinger2002lazy}, where the abstract program is not derived
from a global set of predicates, but is an abstract model in which
predicates change from state to state. Such an abstraction is obtained
through the process of lazy refinement, where abstraction is done
on-the-fly with a goal of eliminating local spurious
counterexamples. While the exact mechanics of our technique are
different, our technique also
generates correctness proofs consisting of lazily generated local
invariants.

\mypara{Logical cuts} In classical logic, a cut serves the role of a
lemma.  In Gentzen's sequent calculus
\cite{GentzenSequentCalculus}, the cut
rule splits the proof tree into two branches, one in which the lemma
can be used as an assumption, and another in which it must be proved.
The cut-elimination theorem, states that
any proof of the sequent calculus that uses the cut rule has another
proof that does not use the cut rule.  Ideas similar to Gentzen's cut
rule have been developed for other reasoning frameworks.  Craig
interpolants \cite{Craig1957} have been used to compute cuts in
frameworks that leverage first-order logic, and they have been used
successfully in a model checking framework \cite{mcmillan2006lazy}.
The differential cut rule of d$\mathcal{L}$ makes it possible to introduce lemmas
about the continuous evolution of differential equations.
It has been shown that there are theorems that cannot be proved without
differential cuts, i.e., the differential cut strictly adds deductive
power
\cite{DBLP:journals/lmcs/Platzer12}.
Overall, the approach provides an iterative method to find
a safety certificate, by proposing sets that are initialized and
invariant, and repeating the differential cut procedure until safety
can be proved.
This work proposes a forward invariant cut rule, in which a lemma is proved about
the evolution of a hybrid system model.  The proof rule requires showing
that a certain set is safe and invariant, and allows the proof to continue
for the behaviors that are not initialized within the set. The forward
invariant
cut may be repeated, until a proof of overall system safety is attained.
Most crucially, the proposed cuts allow the verification process
to leverage a designer's knowledge of local system properties.

\mypara{Deductive Proof System for Temporal Logic} In
\cite{dimitrova2014deductive}, the authors present a deductive proof
system for proving alternating-time temporal logic assertions on a
continuous dynamical system. Some of the proof rules presented require
the user to provide auxilary predicates to establish proof-subgoals.
These predicates are essentially logical cuts, and in particular can
be barrier certificates. The key feature of our approach is that we
provide an automated mechanism to leverage user insight about parts of
the system to obtain localized forward invariant cuts. It would be
interesting to see if the automation that we develop in this paper
could be used to mechanize the proof system presented in
\cite{dimitrova2014deductive}.

\mypara{Conclusions}
This paper presents a method to leverage knowledge of local
system behavior within a deductive framework. In this framework,
designer knowledge of system behavior can be leveraged lazily
as part of a proof of global system safety. The designer
proposes sets that are invariant and safe, which allows
certifying the safety of some region of state space.
In future work, we would like to investigate the use of
sets that are safe, but not initialized or invariant, as part
of a proof effort. An example of this is when a collection of modes
have continuous barriers that the differential equations may not cross,
but the set is not invariant because there are outgoing transitions
that are not excluded by the set.


\bibliographystyle{splncs03}
\bibliography{bibliography}

\begin{thebibliography}{10}
\providecommand{\url}[1]{\texttt{#1}}
\providecommand{\urlprefix}{URL }

\bibitem{PVS2}
Abraham-Mumm, E., Hannemann, U., Steffen, M.: Verification of hybrid systems:
  formalization and proof rules in pvs. In: Engineering of Complex Computer
  Systems, 2001. Proceedings. Seventh IEEE International Conference on. pp.
  48--57 (2001)

\bibitem{PVS1}
Abraham-Mumm, E., Hannemann, U., Steffen, M.: Assertion-based analysis of
  hybrid systems with pvs. In: Moreno-Diaz, R., Buchberger, B., Luis~Freire, J.
  (eds.) Computer Aided Systems Theory, EUROCAST 2001, Lecture Notes in
  Computer Science, vol. 2178, pp. 94--109. Springer Berlin Heidelberg (2001),
  \url{http://dx.doi.org/10.1007/3-540-45654-6_8}

\bibitem{beckert2007verification}
Beckert, B., H{\"a}hnle, R., Schmitt, P.H.: {Verification of object-oriented
  software: The KeY approach}. Springer-Verlag (2007)

\bibitem{boyd94LMI}
Boyd, S., Ghaoui, L.E., Feron, E., Balakrishnan, V.: {Linear Matrix
  Inequalities in System and Control Theory}, vol.~15. SIAM (1994)

\bibitem{branicky1994}
Branicky, M.S.: Stability of switched and hybrid systems. In: Decision and
  Control, 1994., Proceedings of the 33rd IEEE Conference on. vol.~4, pp.
  3498--3503. IEEE (1994)

\bibitem{chen2013}
Chen, X., Abraham, E., Sankaranarayanan, S.: {Flow*: An Analyzer for Non-Linear
  Hybrid Systems}. In: CAV (2013)

\bibitem{Craig1957}
Craig, W.: {Three uses of the Herbrand-Gentzen theorem in relating model theory
  and proof theory}. J. Sym. Logic  3,  269--285 (1957)

\bibitem{dimitrova2014deductive}
Dimitrova, R., Majumdar, R.: Deductive control synthesis for alternating-time
  logics. In: Proc. of the International Conference on Embedded Software. p.~14
  (2014)

\bibitem{Frehse2008phaver}
Frehse, G.: {PHAVer: Algorithmic verification of hybrid systems past HyTech}.
  STTT  10(3),  263--279 (2008)

\bibitem{Frehse2011spaceex}
Frehse, G., Le~Guernic, C., Donz{\'e}, A., Cotton, S., Ray, R., Lebeltel, O.,
  Ripado, R., Girard, A., Dang, T., Maler, O.: {SpaceEx: Scalable verification
  of hybrid systems}. In: CAV. pp. 379--395 (2011)

\bibitem{dReal}
Gao, S., Avigad, J., Clarke, E.M.: $\delta$-complete decision procedures for
  satisfiability over the reals. In: J. Automated Reasoning. pp. 286--300
  (2012)

\bibitem{GentzenSequentCalculus}
Gentzen, G.: {Untersuchungen über das logische Schlie{\ss}en I,II}.
  Mathematische Zeitschrift 39(2), 39(3): 176--210, 405--431 (1935)

\bibitem{henzinger2002lazy}
Henzinger, T.A., Jhala, R., Majumdar, R., Sutre, G.: {Lazy abstraction}. In:
  ACM SIGPLAN Notices. vol.~37, pp. 58--70 (2002)

\bibitem{Henzinger1998}
Henzinger, T.A., Kopke, P.W., Puri, A., Varaiya, P.: {What's Decidable about
  Hybrid Automata?} JCSS  57(1),  94 -- 124 (1998)

\bibitem{Jin14}
Jin, X., Deshmukh, J.V., Kapinski, J., Ueda, K., Butts, K.: {Powertrain Control
  Verification Benchmark}. In: Hybrid Systems: Computation and Control (2014)

\bibitem{JohanssonRantzer1998}
Johansson, M., Rantzer, A.: Computation of piecewise quadratic {Lyapunov}
  functions for hybrid systems. IEEE Tran. on Automatic Control  43(4),
  555--559 (1998)

\bibitem{Kapinski14}
Kapinski, J., Deshmukh, J.V., Sankaranarayanan, S., Ar\'{e}chiga, N.:
  Simulation-guided {Lyapunov} analysis for hybrid dynamical systems. In:
  Hybrid Systems: Computation and Control (2014)

\bibitem{kapinski2004verifying}
Kapinski, J., Krogh, B.H.: Verifying asymptotic bounds for discrete-time
  sliding mode systems with disturbance inputs. In: ACC. vol.~3, pp. 2852--2857
  (2004)

\bibitem{Khalil2002}
Khalil, H.K.: {Nonlinear Systems}. Prentice Hall (2002)

\bibitem{mcmillan2006lazy}
McMillan, K.L.: {Lazy abstraction with interpolants}. In: Computer Aided
  Verification. pp. 123--136 (2006)

\bibitem{Meiss07}
Meiss, J.D.: Differential Dynamical Systems (Monographs on Mathematical
  Modeling and Computation). SIAM (2007)

\bibitem{parrilo2000SDP}
Parrilo, P.A.: {Structured Semidefinite Programs and Semialgebraic Geometry
  Methods in Robustness and Optimization}. Ph.D. thesis, California Institute
  of Technology (2000)

\bibitem{PlatzerLAHS}
Platzer, A.: {Logical Analysis of Hybrid Systems}. Springer (2010)

\bibitem{DBLP:journals/lmcs/Platzer12}
Platzer, A.: The structure of differential invariants and differential cut
  elimination. Logical Methods in Computer Science  8(4),  1--38 (2012)

\bibitem{PlatzerClarke2008}
Platzer, A., Clarke, E.M.: {Computing differential invariants of hybrid systems
  as fixedpoints.} In: Computer Aided Verification (2008)

\bibitem{Prajna2005}
Prajna, S.: {Optimization-based methods for nonlinear and hybrid systems
  verification}. Ph.D. thesis, California Institute of Technology, Caltech,
  Pasadena, CA, USA (2005)

\bibitem{Prajna2006}
Prajna, S.: {Barrier certificates for nonlinear model validation}. Automatica
  42(1),  117--126 (2006)

\bibitem{prajna_safety_2004}
Prajna, S., Jadbabaie, A.: {Safety Verification of Hybrid Systems Using Barrier
  Certificates}. In: Hybrid Systems: Computation and Control. pp. 477--492
  (2004)

\bibitem{topcu08}
Topcu, U., Seiler, P., Packard, A.: {Local stability analysis using simulations
  and sum-of-squares programming}. Automatica  44,  2669--2675 (2008)

\bibitem{Boyd96}
Vandenberghe, L., Boyd, S.: {Semidefinite Programming}. SIAM Review  38(1),
  49--95 (March 1996)

\end{thebibliography}

\newpage
\clearpage

\appendix

\section*{Appendix}\label{sec:appendix}
\renewcommand{\thesubsection}{\Alph{subsection}}
\subsection{Semantics of d$\mathcal{L}$}
We follow the development of \cite{PlatzerLAHS}, Chapter 2.
Symbols in d$\mathcal{L}$ are classified into three
different syntactic categories, depending on their role.
\begin{enumerate}
\item $\Sigma_r$ represents a set of {\em rigid} symbols
that cannot change their value, such as $0, 1, +, \cdot$;
\item $\Sigma_{fl}$ represents a set of {\em flexible} symbols,
also called {\em state variables}, which change their value
as the system evolves;
\item $V$ represents a set of {\em logical variables}, which
do not change as the system evolves, but can be quantified over
universally and existentially; they often serve the role of
parameters.
\end{enumerate}
An {\em interpretation} is a function
$\mathcal{I}$ that associates
functions and relations over the reals
to function and relation symbols
in $\Sigma_r$. The standard arithmetic operators and relations
symbols,
such as $+$, $\cdot$, $\ge$,
are interpreted as usual.
A \emph{state} is a map $\nu: \Sigma_{fl} \mapsto \mathbb{R}$,
which maps a real value to each state variable.
An assignment $\eta: V \mapsto \mathbb{R}$ is a map that
prescribes the value of the logical variables. Note that
the value of the logical variables does not depend on the
state. 

A state variable is a {\em term}, and a logical variable is also 
a term.
The result of applying a function of arity $n$ to $n$ terms
is also a term. Nothing else is a term.

\begin{definition}[Valuation of terms (\cite{PlatzerLAHS}, Defn. 2.5)]
The \emph{valuation of terms} with respect to interpretation
$\mathcal{I}$, assignment $\eta$, and state $\nu$ is defined as
\begin{enumerate}
\item $val_{\mathcal{I}, \eta}(\nu, p) = \eta(p)$ if $p$ is a 
logical variable.
\item $val_{\mathcal{I}, \eta}(\nu, x) = \nu(x)$ if $x$ is a state
variable.
\item $val_{\mathcal{I}, \eta}(\nu, f(\theta_1, \dots, \theta_n)) = 
\mathcal{I}(f)( val_{\mathcal{I}, \eta}(\theta_1), \dots,
val_{\mathcal{I}, \eta}(\theta_2))$ if $f$ is a function of arity
$n \ge 0$ and $\theta_1, \dots, \theta_n$ are terms.
\end{enumerate}
\end{definition}

The notation $\eta[x \mapsto d]$ represents the function that
agrees with $\eta$ except for the interpretation of $x$, where it takes
the value $d$.
The notation $\nu[x \mapsto d]$ denotes the modification of a state
$\nu$, that agrees with $\nu$ everywhere except the interpretation
of the state variable $x$, where it takes the value $d$.

\begin{definition}[Valuation of d$\mathcal{L}$ formulas (\cite{PlatzerLAHS}, Defn. 2.6]
The valuation $val_{\mathcal{I}, \eta}(\nu, \cdot)$
of formulas with respect to interpretation $\mathcal{I}$,
assignment $\eta$, and state $\nu$ is defined as
\begin{enumerate}
\item $val_{\mathcal{I}, \eta}(\nu, p(\theta_1, \dots,\theta_n ) 
= I(p)(val_{\mathcal{I}, \eta}(\nu, \theta_1), \dots, val_{\mathcal{I}, \eta}(\nu, \theta_n))$.
\item $val_{\mathcal{I}, \eta}(\nu, \phi~\wedge~\psi) = {\tt true}$ iff $val_{\mathcal{I}, \eta}(\nu,\phi) = {\tt true} $
and $val_{\mathcal{I}, \eta}(\nu,\psi) = {\tt true} $.
\item $val_{\mathcal{I}, \eta}(\nu, \phi~\vee~\psi ) = {\tt true}$ iff $val_{\mathcal{I}, \eta}(\nu,\phi) = {\tt true} $
or $val_{\mathcal{I}, \eta}(\nu,\psi) = {\tt true} $.
\item $val_{\mathcal{I}, \eta}(\nu, \neg \phi) = {\tt true}$ iff $val_{\mathcal{I}, \eta}(\nu,\phi) \neq {\tt true} $.
\item $val_{\mathcal{I}, \eta}(\nu, \phi \rightarrow \psi) = {\tt true}$ iff $val_{\mathcal{I}, \eta}(\nu,\phi) \neq {\tt true} $
or $val_{\mathcal{I}, \eta}(\nu,\psi) = {\tt true} $.
\item $val_{\mathcal{I}, \eta}(\nu, \forall x \phi ) = {\tt true}$ iff $val_{\mathcal{I}, \eta_{ [x \mapsto d]}}(\nu, \phi) = {\tt true}$
for all $d \in \mathbb{R}$.
\item $val_{\mathcal{I}, \eta}(\nu, \exists x \phi ) = {\tt true}$ iff $val_{\mathcal{I}, \eta_{ [x \mapsto d]}}(\nu, \phi) = {\tt true}$
for some $d \in \mathbb{R}$.
\item $val_{\mathcal{I}, \eta}(\nu, [\alpha] \phi ) = {\tt true}$ 
iff $val_{\mathcal{I}, \eta}(\omega, \phi) = {\tt true}$
for all states $\omega$ for which the 
transition relation (defined below) satisfies $(\nu, \omega) \in \rho_{\mathcal{I}, \eta}(\alpha)$.
\item $val_{\mathcal{I}, \eta}(\nu, \left< \alpha \right> \phi ) = {\tt true}$ iff
$val_{\mathcal{I}, \eta}(\omega, \phi)$
for some state $\omega$ such that the transition relation satisfies $(\nu, \omega) \in \rho_{\mathcal{I}, \eta}(\alpha)$.
\end{enumerate}
\end{definition}
 
We now define the transition semantics of hybrid programs. We already saw a glimpse of it in the definition of valuation of formulas,
since the formulas and programs of d$\mathcal{L}$ are constructed coinductively.

\begin{definition}[Transition semantics of hybrid programs (\cite{PlatzerLAHS}, Defn. 2.7)]
The valuation of a hybrid program $\alpha$, denoted $\rho_{\mathcal{I},\eta}(\alpha)$ is a transition relation on states
that specifies which states are reachable from a state $\nu$ under the program $\alpha$, and is defined inductively as follows.
\begin{enumerate}
\item $(\nu, \omega) \in \rho_{\mathcal{I}, \eta}(x_1 \coloneqq \theta_1, \dots, x_n \coloneqq \theta_n )$ iff
the state $\omega$ equals the state obtained by modification of $\nu$ as
$\nu[x_1 \mapsto val_{\mathcal{I},\eta}(\nu, \theta_1)], \dots, \nu[x_n \mapsto val_{\mathcal{I},\eta}(\nu, \theta_n)]$.
\item $(\nu, \omega) \in \rho_{\mathcal{I}, \eta}(\{x_1' = \theta_1, \dots, x_n' = \theta_n \& H \} )$ iff there is a flow $f$ of some
duration $r \ge 0$ from $\nu$ to $\omega$ along the differential equations $x_1' = \theta_1, \dots, x_n' = \theta_n$ that always
respects the invariant $H$.
\item $\rho_{\mathcal{I}, \eta}(? \chi ) = \{(\nu, \nu)~|~val_{\mathcal{I}, \eta}(\nu, \chi) = {\tt true} \}$
\item $\rho_{\mathcal{I}, \eta}(\alpha \cup \beta) =  \rho_{\mathcal{I}, \eta}(\alpha) \cup  \rho_{\mathcal{I}, \eta}(\beta)$
\item $\rho_{\mathcal{I}, \eta}(\alpha; \beta ) = \rho_{\mathcal{I}, \eta}(\alpha) \circ  \rho_{\mathcal{I}, \eta}(\beta) $
\item $(\nu, \omega) \in \rho_{\mathcal{I}, \eta}(\alpha^*)$ iff there is a sequence
of states states $\nu_0,\dots,\nu_n$ with $n \ge 0$, $\nu = \nu_0$, and $\nu_n = \omega$ such that 
$(\nu_i, \nu_{i+1}) \in \rho_{\mathcal{I}, \eta}(\alpha)$ for each $0 \le i \le n-1$.
\end{enumerate}
\end{definition}

\subsection{Soundness proof for forward invariant cut}
Fix an interpretation $\mathcal{I}$ and an assignment $\eta$.
From semantics of the first premise, if $\nu \in C$ and $(\nu, \omega) \in \rho_{\mathcal{I},\eta}(\alpha)$,
then $\omega \in C$.
From the semantics of the second premise, if $\omega \in C$, then $\omega \in S$
From the semantics of the third premise, if $\nu \in I$ and $\nu \notin C$, and $\omega$
is such that  $(\nu, \omega) \in \rho_{\mathcal{I},\eta}( (\alpha; ?\neg C)^* )$, then $\omega \in S$.
This is equivalent to saying that for any $\omega$ such that there is a sequence of states $\nu_0, \dots, \nu_n$, with $\nu_0 = \nu \in I$
and $\nu_n = \omega$, $n \in \mathbb{N}$,
and $(\nu_i, \nu_{i+1}) \in \rho_{\mathcal{I},\eta}(\alpha; ?\neg C )$ for each $0 \le i \le n-1$, it is the case that
$\omega \in S$.

The proof is to show by induction that any state reachable by $\alpha^*$ from $I$ in $n \ge 0$ executions of $\alpha$
must be contained in
$S$.

For the base case, let $n = 0$. Then given $\nu \in I$, the only reachable state by a sequence of length zero is $\nu$
itself. If $\nu \in C$, then $\nu in S$ by semantics of the second premise. If $\nu \notin C$,
we have that $(\nu, \nu) \in \rho_{\mathcal{I},\eta}( (\alpha; ?\neg C)^* )$ by a chain of length zero,
so that by semantics of the third premise, $\nu \in S$.

As an inductive hypothesis, suppose that for every $\omega$ reachable by a chain of length $n$, $\omega \in S$
 (\ie, there exists $\nu_0, \dots, \nu_n$
with $\nu_0 = \nu$ and $\omega = \nu_n$ such that $(\nu_i, \nu_{i+1}) \in \rho_{\mathcal{I},\eta}(\alpha)$, for $0 \le i \le n-1$.
Now choose any state $\xi$ such that there is a chain of length $n+1$, $\nu_0, \dots, \nu_{n+1}$ with $\nu_0 = \nu$
and $\nu_{n+1}=\xi$, such that  $(\nu_i, \nu_{i+1}) \in \rho_{\mathcal{I},\eta}(\alpha)$, for $0 \le i \le n$).

First suppose that $\nu_n \in C$. Then by semantics of the first premise, $\nu_{n+1} \in C$,
and then $\nu_{n+1} \in S$ by semantics of the second premise.
On the other hand, suppose $\nu_n \notin C$. We claim that for all $j \le n$, $\nu_j \notin C$. To see this, note
that if $\nu_j \in C$ for some $j \le n$, then $\nu_n \in C$ by semantics of the first premise, which
would contradict our assumption on $\nu_n$. Then we have that $(\nu_i, \nu_{i+1}) \in \rho_{\mathcal{I},\eta}( \alpha; ?\neg C )$
for all $0 \le i \le n$. By semantics of the third premise, it follows that $\xi \in S$.
This establishes the theorem.

\subsection{System dynamics for the Engine Fuel Control Model}

We now present the model parameters and the ODEs for the Engine Fuel
Control model. Figure \ref{fig:odestartup} details the equations for the $\startup$ mode, and Fig. \ref{fig:odes} provides the dynamic equations for the $\normalmode$ mode. In the figures, $\frac{dp}{dt}=f_p$, $\frac{dr}{dt}=f_r$, $\frac{d\pest}{dt}=f_{\pest}$, and $\frac{di}{dt}=f_i$. 

\begin{figure*}[t]

\begin{equation}\nonumber
\begin{array}{lll}
\vspace{0.6em}
f_p  & =  & \displaystyle 
                c_1 \left(  2\hat{u_1}\sqrt{\frac{p}{c_{11}}-\left(\frac{p}{c_{11}} \right)^2} - 
                            \left(c_3+c_4c_2p  + c_5 c_2 p^2 + c_6 c_2^2 p \right)
                    \right)  \\
\vspace{0.6em}
f_r & =  &  \displaystyle 
	                4\left( \frac{c_3 + c_4 c_2 p + c_5c_2 p^2 + c_6 c_2^2 p}
                             {c_{13} (c_3 + c_4 c_2 \pest^2 + c_5 c_2 \pest^2 + c_6 c_2^2 \pest)}
                        - r
                  \right)  \\
f_{\pest} & = & c_1 \left(2 \hat{u_1}\sqrt{\frac{p}{c_{11}} - \left(\frac{p}{c_{11}}\right)^2} - 
                            c_{13}\left(c_3 + c_4 c_2 \pest + c_5 c_2 \pest^2 + c_6 c_2^2 \pest\right)
                      \right) \\
f_i  & =  & 0  
\end{array}
\end{equation}

\caption{System dynamics for the Engine Fuel Control System in the $\startup$ mode.\label{fig:odestartup}}
\end{figure*}

\begin{figure*}[t]

\begin{equation}\nonumber
\begin{array}{lll}
\vspace{0.6em}
f_p  & =  & \displaystyle 
                c_1 \left(  2\hat{u_1}\sqrt{\frac{p}{c_{11}}-\left(\frac{p}{c_{11}} \right)^2} - 
                            \left(c_3+c_4c_2p  + c_5 c_2 p^2 + c_6 c_2^2 p \right)
                    \right)  \\
\vspace{0.6em}
f_r  & =  &  \displaystyle 
	                4\left( \frac{c_3 + c_4 c_2 p + c_5c_2 p^2 + c_6 c_2^2 p}
                             {c_{13} (c_3 + c_4 c_2 \pest^2 + c_5 c_2 \pest^2 + c_6 c_2^2 \pest)(1 + i + c_{14}(r - c_{16}))}
                        - r
                  \right)  \\
f_{\pest} & = & c_1 \left(2 \hat{u_1}\sqrt{\frac{p}{c_{11}} - \left(\frac{p}{c_{11}}\right)^2} - 
                            c_{13}\left(c_3 + c_4 c_2 \pest + c_5 c_2 \pest^2 + c_6 c_2^2 \pest\right)
                      \right) \\
f_i  & =  & c_{15}(r-c_{16})  
\end{array}
\end{equation}

\caption{System dynamics for the Engine Fuel Control System in the $\normalmode$ mode.\label{fig:odes}}
\end{figure*}

{\footnotesize
\begin{table}[t]
\centering
\caption{Model Parameters for the Engine Fuel Control System.\label{table:parameters}} 
\begin{tabular}{|c|c|}
\hline Parameter & Value \\ \hline 
$c_1$ & $0.41328$  \\
$c_2$ & $200.0$   \\
$c_3$ & $-0.366$  \\
$c_4$ & $0.08979$ \\
$c_5$ & $-0.0337$ \\
$c_6$ & $0.0001$  \\
$c_7$ & $2.821$   \\ 
$c_8$ & $-0.05231$ \\
$c_9$ & $0.10299$ \\
$c_{10}$ & $-0.00063$  \\
$c_{11}$ & $1.0$  \\
$c_{12}$ & $14.7$ \\
$c_{13}$ & $0.9$  \\
$c_{14}$ & $0.4$  \\
$c_{15}$ & $0.4$  \\
$c_{16}$ & $1.0$  \\
$\hat{u_1}$ & $23.0829$ \\ 
\hline
\end{tabular}
\end{table}}

We translate the system so that the origin coincides with the $\normalmode$
equilibrium point $p\approx 0.8987$, $r=1.0$, $\pest \approx 1.077$,
$i\approx 0.0$ and call the translated variables $\phat$, $\rhat$,
$\pesthat$, and, $\ihat$, respectively.

\end{document}